\documentclass[12pt, draftclsnofoot, onecolumn]{IEEEtran}
\usepackage{hyperref}
\usepackage{amsmath,amssymb}
\usepackage{latexsym}
\usepackage{CJK}
\usepackage{cite}   

\usepackage{verbatim}  
\usepackage{graphicx}  
\usepackage{mathrsfs}	
\usepackage{algorithm} 
\usepackage{algpseudocode}
\usepackage{booktabs}
\usepackage{textcomp}  
\usepackage{multirow}  
\usepackage{lettrine}   
\usepackage{graphicx}  
\usepackage{amsthm} 
\usepackage{color}  
\newtheorem{thm}{Theorem}
\newtheorem{cor}{Corollary}
\newtheorem{lem}{Lemma}

\renewcommand{\algorithmicrequire}{\textbf{Input:}}
\renewcommand\algorithmicensure {\textbf{Output:} }
\renewcommand\algorithmicrepeat {\textbf{Repeat:} }
\renewcommand\algorithmicuntil {\textbf{Until:} }

\begin{document}

\title{Subcarrier Grouping with Environmental Sensing for MIMO-OFDM Systems over Correlated Double-selective Fading Channels}

\author{\IEEEauthorblockN{Jiaxun Lu, Zhengchuan Chen, Pingyi Fan, Khaled B. Letaief}

\IEEEauthorblockA{
Tsinghua National Laboratory for Information Science and Technology(TNList),\\
 Department of Electronic Engineering, Tsinghua University,Beijing,China.\\
 Department of Electronic and Computer Engineering, Hong Kong University of Science and Technology, Hong Kong\\
  Email: \{lujx14,chenzc10\}@mails.tsinghua.edu.cn, fpy@mail.tsinghua.edu.cn, eekhaled@ust.hk}
}

\maketitle

\graphicspath{{Figures/}}

\begin{abstract}

MIMO-OFDM is a promising technique in 5G wireless communications. In high mobility scenarios, the transmission environments are time varying and/or the relative moving velocity between the transmitter and receiver is also time varying. In the literature, most of previous works mainly focused on fixed subcarrier group size and precoded the MIMO signals with unitary channel state information (CSI). In this way, the subcarrier grouping may naturally lead to big loss of channel capacity in high mobility scenarios due to the CSI difference on the subcarriers in each group. To employ the MIMO-OFDM technique, adaptive subcarrier grouping scheme may be an efficient way. In this paper, we first consider MIMO-OFDM systems over double-selective i.i.d. Rayleigh channels and investigate the quantitative relation between subcarrier group size and capacity loss theoretically. With developed theoretical results, we also propose an adaptive subcarrier grouping scheme to satisfy the preset capacity loss threshold by adjusting grouping size with the sensed environmental information and mobile velocity. Theoretical analysis and simulation results show that to achieve a better system capacity, a sparse scattering, lower SNR and lower velocity as well as properly large antenna number are matched with larger subcarrier group size. One important observation is that if the antenna number is too large and higher than a threshold, which will not bring any additional gain to the subcarrier grouping. That is, the system capacity loss will converge to a lower-bound expeditiously with respect to antenna number, which is given in theory also.

\end{abstract}


\IEEEpeerreviewmaketitle

\section{Introduction}\label{Sec:Introduction}

\lettrine[lines=2]{H}{igh} mobility is one of the  key scenarios to 5G wireless communications, in which, the transmission environments are time varying and/or the relative moving velocity between the  transmitter and receiver  is also time varying. Thus, how to keep high spectrum efficiency and  reliable transmission performance in such a scenarios is becoming a full of challenging topic. As we know, the OFDM modulation scheme can provide high spectrum efficiency and enhance system performance in frequency selective fading channels \cite{stuber2004broadband}, { {while the multi-antenna structure improves system capacity, especially for the massive MIMO regimes \cite{larsson2014massive}. Also, MIMO regimes enhances system reliability as well as improves coverage \cite{sellathurai2010mimo,paulraj2003introduction}.}} Therefore, MIMO-OFDM is a promising technique in 5G systems, due to its provision capability of high capacity and spectrum efficiency in rich scattering wireless channels { \cite{you2015adjustable}}. But most of the joint design works on MIMO-OFDM are assumed that its working scenarios are relatively qusi-static. Thus, the transmission scheme of MIMO-OFDM in mobility scenarios needs to be carefully designed. That is, it should be capable of environmental sensing and adaptively adjusting to varying transmit environment and moving velocity.


In MIMO-OFDM regimes, various methods like frequency and spatial water-filling are performed to maximize the capacity of wireless channels\cite{tse2005fundamentals}. { {Also, \cite{joung2015multicast} proposes a linear precoding scheme to enhance the minimum user rate on multicast MIMO-OFDM systems.}} In addition, for convenience of symbol modulation and signal decoding at the transmit and receive side, linear or nonlinear signal precoding schemes on each individual subcarrier are carried out \cite{vu2007mimo,jiang2005joint,jiang2005uniform}. For instance in \cite{jiang2005uniform}, the signal precoding scheme decomposes MIMO channels into several uniform sub-channels, which simplifies subsequent modulation and improve diversity gains of system significantly. However, either water-filling or signal precoding demands perfect channel state information at the transmitter (CSIT). When CSIT is not readily available, it is often possible to use a feedback channel to provide channel state information (CSI) to the transmitter or just estimate CSI at the transmit side via reciprocal principle when time division duplex is used. The procedure in achieving CSIT and obtaining signal precoding matrices may occupy much too system resources, either frequency bandwidth or computational complexity. Especially in high subcarrier or antenna number regimes \cite{edfors2014massive,pande2007reduced}.

An effective scheme to reduce system cost is to precode with subcarrier grouping, which is not a new idea { {\cite{pande2007reduced,choi2006interpolation,zhang2007reduced,lin2015feedback}.}} In subcarrier grouping regimes, transmission system will group subcarriers via a specific scheme and apply an unitary precoding matrix to the subcarriers in one group. Thus, one can linearly reduce the computational complexity of precoding process as well as decrease the system feedback CSI amount. However, the subcarrier grouping scheme naturally leads to system capacity loss caused by the difference of CSIT at subcarriers in one subcarrier group, the larger the subcarrier group is, the greater the capacity loss suffers. This means that there exists a tradeoff between the subcarrier group size and system channel capacity. However, previous works { {in \cite{pande2007reduced,choi2006interpolation,zhang2007reduced,lin2015feedback}}} focused on the process of feedback information, but neglected the adaptive adjustment of subcarrier group size. Hence, previous precoding schemes with fixed grouping size may not keep an effective adaption of system services performance to varying double-selective channels. { {\cite{ge2014energy} proposes an adaptive subchannel grouping scheme to simplify the multichannel optimization problem, while satisfying system services performance constraint. Similarly, \cite{lu2014precoding} discussed the adaptive subcarrier grouping method only in frequency-selective channels.}} Now we will extend it to double-selective channels in this work.

In this paper, we firstly investigate how the difference of CSIT affects channel capacity loss. Then, by exploiting the relations between channel correlation and the second order of CSIT difference, we derive quantitative relations between channel capacity loss and subcarrier group size. The quantitative relations reveal the explicit tradeoff between system service capacity and system cost. In addition, based on the new obtained theoretical results, we also propose an adaptive subcarrier grouping based MIMO-OFDM signal precoding algorithm to satisfy the preset capacity loss threshold, which adjusts the subcarrier group size adaptively with the sensed transmit environment information and mobile velocity.

 The rest of this paper is organized as follows. Section \uppercase\expandafter{\romannumeral2} presents the system model and its scenario of the subcarrier grouping based precoding scheme. In this section, we explain what the CSI difference is and its relations with CSIT difference. Section \uppercase\expandafter{\romannumeral3} derives the channel capacity with CSI difference in the view of information theory. Meanwhile, we also discuss the approximation of corresponding ergodic capacity in this section. In Section \uppercase\expandafter{\romannumeral4}, the relations between CSI difference and channel correlation are investigated. Then, the subcarrier grouping method, given the capacity loss constraint, is established with a mapping method. Based on the obtained theoretical results, we propose an adaptive subcarrier grouping based precoding algorithm in Section \uppercase\expandafter{\romannumeral4} and give the theoretical performance analysis in various scenarios. Section \uppercase\expandafter{\romannumeral5} examines the effectiveness and application conditions of proposed subcarrier grouping scheme. Conclusions of this work are given in Section \uppercase\expandafter{\romannumeral6}.

\emph{Notation}: $(\cdot)^{-1}$ and $(\cdot)^{\dagger}$ denote the inverse and conjugate transpose of $(\cdot)$, respectively. The symbols $\mathbf{0}$ and $\mathbf{I}$ denote the zero and the identical matrix, respectively.

\section{System Model And Problem Formulation}\label{Sec:Sysmodel}

In this section, we shall first review the system model of MIMO-OFDM system. Then, we explain how a subcarrier grouping based MIMO-OFDM signal precoding algorithm works and introduce how the CSI difference caused by subcarrier grouping schemes affects the MIMO-OFDM communication systems.

\subsection{Review of System Structure}

\begin{figure*}[htbp]
\centering
\includegraphics[width=0.8\textwidth]{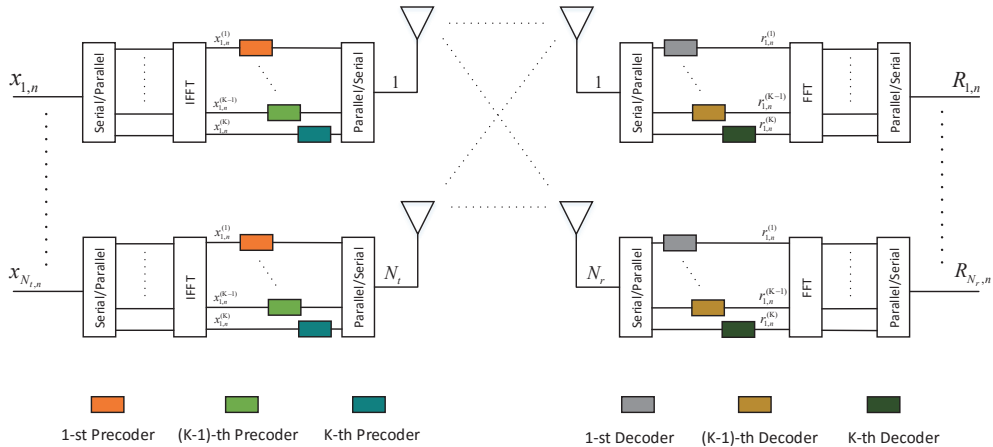}
\caption{Subcarrier grouping based MIMO-OFDM signal precoding system.} \label{Fig:SystemDiagram}
\end{figure*}

Let us consider a MIMO-OFDM communication system as shown in Fig. \ref{Fig:SystemDiagram}. It involves with $N_t$ transmit antennas and $N_r$ receive antennas. $K$ subcarriers are taken into account in each OFDM blocks. For the concise of analysis, cyclic prefixes are ignored and let the $(n,k)$-th subcarrier denote the $k$-th subcarrier in the $n$-th OFDM block. Denote the transmit symbol vector at the $(n,k)$-th subcarrier as $\mathbf{X}_{n}^{(k)} = [x_{1,n}^{(k)},x_{2,n}^{(k)},\cdots,x_{N_t,n}^{(k)}]^{T}$. The corresponding $N_r \times N_t$ channel fading matrix is $\mathbf{H}_{n}^{(k)}$ and the $N_t \times N_t$ precoding matrix is $\mathbf{F}_{n}^{(k)}$. Thus, the precoded transmit symbol vector can be written as

\begin{equation}\label{Equ:PrecodedSignal}
\mathbf{\widetilde{X}}_{n}^{(k)} = \mathbf{F}_{n}^{(k)} \mathbf{X}_{n}^{(k)}.
\end{equation}

At the receiver, the corresponding received symbol vector $\mathbf{R}_{n}^{(k)}$ can be obtained by

\begin{equation}\label{Equ:Rk=HS+N}
\begin{split}
\mathbf{R}_{n}^{(k)} = \mathbf{H}_{n}^{(k)} \mathbf{\widetilde{X}}_{n}^{(k)} + \mathbf{W}_{n}^{(k)}
= \mathbf{H}_{n}^{(k)} \mathbf{F}_{n}^{(k)} \mathbf{X}_{n}^{(k)} + \mathbf{W}_{n}^{(k)},
\end{split}
\end{equation}
where $\mathbf{R}_{n}^{(k)} = [r_{1,n}^{(k)},r_{2,n}^{(k)},\cdots,r_{N_t,n}^{(k)}]^{T} $ and $\mathbf{W}_{n}^{(k)}$ is the normalized circular complex additive white Gaussian noise (AWGN). In FDD regimes, the receiver needs to feed back the CSI, i.e. $\mathbf{H}_{n}^{(k)}$, to transmitter, while in TDD regimes it can obtain the CSI via the reciprocal principle. The obtaining of CSI may acquire extra feedback bandwidth or computational complexity cost, which may decrease system service performance. Especially for the fast time varying channels, the computational complexity and  the overhead cost for channel estimation can be unsustainable in practice.

\subsection{Subcarrier Grouping And Channel Model with CSI Difference}

In this part, we shall introduce the subcarrier grouping method and the corresponding channel model with CSI difference. Let $(n_c,k_c)$ be the center subcarrier. The corresponding subcarrier group is denoted as $\pounds(n_c,k_c)$, which is composed by the adjacent subcarriers of $(n_c,k_c)$ via specific method. For all subcarriers with subscript $(n,k) \in \pounds(n_c,k_c)$, one can apply precoding matrix $\mathbf{F}_{n_c}^{(k_c)}$ to precode transmit vectors at the $(n,k)$-th subcarrier. Thus, the subcarrier grouping method can be modeled by

\begin{equation}\label{Equ:subcarrierGrouping}
\mathbf{R}_{n}^{(k)} = \mathbf{H}_{n}^{(k)} \mathbf{F}_{n_c}^{(k_c)} \mathbf{X}_{n}^{(k)} + \mathbf{W}_{n}^{(k)}.
\end{equation}

Therefore, only the CSI and precoding matrix at the $(n_c,k_c)$-th subcarrier needs to be acquired and generated.  As seen in \eqref{Equ:subcarrierGrouping}, it shows that subcarrier grouping scheme naturally leads to the CSI difference between the $(n,k)$- and $(n_c,k_c)$-th subcarriers, defined as

\begin{equation}\label{Equ:ChannelMismatch}
\mathbf{M}_{n,n_c}^{(k,k_c)} = \mathbf{H}_{n}^{(k)} - \mathbf{H}_{n_c}^{(k_c)}.
\end{equation}

Obviously, the CSI difference introduces interference to systems and reduce channel capacity. This is the cost by using subcarrier grouping technique, which limits the application of subcarrier grouping in certain degree. That is, there exists a tradeoff between subcarrier group size and channel capacity loss. We shall discuss it in the sequel.

 Substituting \eqref{Equ:ChannelMismatch} into \eqref{Equ:subcarrierGrouping}, the subcarrier grouping model can be rewritten as

\begin{equation}\label{Equ:CapacityModel}
\begin{split}
\mathbf{R}_{n}^{(k)} =& (\mathbf{H}_{n_c}^{(k_c)} + \mathbf{M}_{n,n_c}^{(k,k_c)})\mathbf{F}_{n_c}^{(k_c)}\mathbf{X}_{n}^{(k)} + \mathbf{W}_{n}^{(k)}\\
=& \mathbf{H}_{n_c}^{(k_c)}\mathbf{F}_{n_c}^{(k_c)}\mathbf{X}_{n}^{(k)} + \underbrace{ \mathbf{M}_{n,n_c}^{(k,k_c)} \mathbf{F}_{n_c}^{(k_c)}\mathbf{X}_{n}^{(k)} + \mathbf{W}_{n}^{(k)} }_{\mathbf{N}_{n,n_c}^{(k,k_c)}},
\end{split}
\end{equation}
where $\mathbf{M}_{n,n_c}^{(k,k_c)}\mathbf{F}_{n_c}^{(k_c)}\mathbf{X}_{n}^{(k)}$ is the additional interference introduced by CSI difference. Denote the total interference including the additional interference and AWGN as the equivalent interference $\mathbf{N}_{n,n_c}^{(k,k_c)}$. It can be proved that it's related to channel self-correlation, based on which, the quantitive relation between capacity loss and subcarrier group size will be  shown in Section III.

Later, we present an algorithm to satisfy the given capacity loss threshold by  adaptively adjusting the size of subcarrier groups with the concurrent transmit environment and mobile velocity. That is, we shall investigate the quantitative relationship between subcarrier group size and capacity loss and design a feasible subcarrier grouping algorithm to adaptively track the sensed channel time varying.

\section{Capacity Analysis}\label{Sec:Capacity}

In this section, we derive the capacity loss of each individual subcarrier in terms of CSI difference. We first present a result on capacity loss under the scenario with CSI difference and then analyze the approximate ergodic capacity when the subcarrier grouping based  precoding algorithm is employed.

\subsection{Capacity Lower Bound With CSI difference}

Considering the fact, over a wide range of SNR, the gain of water-filling is very small, especially in high SNR. In this case, allocating equal power at all sub-channels is almost optimal \cite{tse2005fundamentals}. Thus, we assume the equal power allocation scheme over both the frequency and spatial domains are adopted in our considered systems. Hence, the design of the precoding matrix $\mathbf{F}_{n_c}^{(k_c)}$ can degrade into the search of unitary matrix\cite{jiang2005uniform}. It is well known that the unitary precoding matrix doesn't change the distribution of random variables. Based on this result, we can rewrite \eqref{Equ:CapacityModel} in an equivalent form as

\begin{equation}\label{Equ:CapacityModSimp}
\begin{split}
\mathbf{R}_{n}^{(k)} =& \left(\mathbf{H}_{n_c}^{(k_c)} + \mathbf{M}_{n,n_c}^{(k,k_c)}\right)\mathbf{X}_{n}^{(k)} + \mathbf{W}_{n}^{(k)}\\
=& \mathbf{H}_{n_c}^{(k_c)}\mathbf{X}_{n}^{(k)} + \underbrace{ \mathbf{M}_{n,n_c}^{(k,k_c)} \mathbf{X}_{n}^{(k)} + \mathbf{W}_{n}^{(k)} }_{\mathbf{N}_{n,n_c}^{(k,k_c)}},
\end{split}
\end{equation}
where the precoding matrix $\mathbf{F}_{n_c}^{(k_c)}$ has been absorbed into the random vectors $\mathbf{X}_{n}^{(k)}$.

In the channel fading matrix $\mathbf{H}_{n}^{(k)}$, the $i$-th row represents the fading factor from the $i$-th receive antenna to $N_t$ transmitters at the $(n,k)$-th subcarrier, which will simply be denoted by $\mathbf{g}_{i,n}^{(k)}$ for subsequent use. Similarly, we denote the $j$-th column of $\mathbf{H}_{n}^{(k)}$ as $\mathbf{h}_{j,n}^{(k)}$, which represents the fading factor from the $i$-th transmit antenna to $N_r$ receivers at the $(n,k)$-th subcarrier. Denote

\begin{equation}\label{Equ:definationOfmn}
u = \max\{N_t , N_r\}, ~~~ v = \min\{N_t , N_r\}.
\end{equation}

Given the channel fading matrix $\mathbf{H}_{n}^{(k)}$, the capacity of the $(n,k)$-th subcarrier without CSI difference (i.e. $\mathbf{M}_{n,n_c}^{(k,k_c)} = \mathbf{0}$) is

\begin{equation}\label{Equ:CapacityWithoutMismatch}
C_{n}^{(k)} = {\log}_2 ~\det\left(\mathbf{I} + \gamma \mathbf{S}_{n}^{(k)}\right),
\end{equation}
where the $v\times v$ matrix $\mathbf{S}_{n}^{(k)}$ is defined by
\begin{equation}\label{Equ:definationS}
\mathbf{S}_{n}^{(k)} =
\begin{cases}
\mathbf{H}_{n}^{(k)}\mathbf{H}_{n}^{(k)\dagger} = \begin{matrix} \sum_{i=1}^u {\mathbf{h}_{i,n}^{(k)}}{\mathbf{h}_{i,n}^{(k)}}^\dagger \end{matrix}, & \mbox{if } N_t \geq N_r \\
\mathbf{H}_{n}^{(k)\dagger} \mathbf{H}_{n}^{(k)} = \begin{matrix} \sum_{i=1}^u {\mathbf{g}_{i,n}^{(k)}}^\dagger{\mathbf{g}_{i,n}^{(k)}} \end{matrix}, & \mbox{if } N_t < N_r.
\end{cases}
\end{equation}

For simplifying the calculation, we assume that $\mathbf{S}_{n}^{(k)}$ is normalized so that $\rm{trace}\{E[\mathbf{S}_{n}^{(k)}]\}/v = 1$, where $\mathrm{E}(\cdot)$ represents the expectation notation. Accordingly, the SNR at each subcarrier on one antenna is $\gamma = \rho_t / N_t$, where $\rho_t$ is the total transmit power. As to the scenarios with CSI difference (i.e. $\mathbf{M}_{n,n_c}^{(k,k_c)} \neq 0$), the capacity loss varies with $\mathbf{M}_{n,n_c}^{(k,k_c)}$. Thus, one cannot adjust transmit rate without exact CSI, which cannot be acquired in this subcarrier grouping scheme. For sake of transmission reliability, as an alternative, we use the following lower-bound instead.

\begin{equation}\label{Equ:BoundsOfMismatch}
C_{n}^{(k)} \geq  {\log}_2~\det\left(\mathbf{I} + \gamma _e \mathbf{S}_{n_c}^{(k_c)}\right),
\end{equation}
where $\gamma _e$ is  the equivalent-SNR (ESNR) in the worst case, which is given by

\begin{equation}\label{Equ:SNROfCapacityBounds}
\begin{split}
\gamma _e = & \sigma_{\widetilde {\mathbf{X}}_n^{(k)}}^2 {\sigma_{\mathbf{N}_{n,n_c}^{(k,k_c)}}^2}^{-1}= \sigma_{\mathbf{X}_n^{(k)}}^2 {\sigma_{\mathbf{N}_{n,n_c}^{(k,k_c)}}^2}^{-1} \\
=& {\sigma _{x_{n}^{(k)}} ^2 }/({v \sigma_{m_{n,n_c}^{(k,k_c)}} ^2 \sigma _{x_{n}^{(k)}} ^2 + \sigma_{w_{n}^{(k)}}^2}).
\end{split}
\end{equation}

The variables $\sigma _{x_{n}^{(k)}} ^2$, $\sigma_{m_{n,n_c}^{(k,k_c)}} ^2$ and $\sigma_{w_{n}^{(k)}}^2$ are the second-moment of the elements of $\mathbf{X}_{n}^{(k)}$, $\mathbf{M}_{n,n_c}^{(k,k_c)}$ and $\mathbf{W}_{n}^{(k)}$, respectively. Appendix \ref{Appen:capacityThm} gives a brief derivation of this lower-bound.

It's clearly seen from \eqref{Equ:BoundsOfMismatch} that $\mathbf{M}_{n,n_c}^{(k,k_c)}$  is totally contributed to interference in the worst case. Combining \eqref{Equ:CapacityWithoutMismatch} with \eqref{Equ:BoundsOfMismatch}, the upper-bound of capacity loss caused by CSI difference can be described via the following lemma.

\begin{lem}\label{Lem:CapacityLossUpperbound}
 Suppose  the center subcarrier is the $(n_c,k_c)$-th subcarrier. The upper-bound of the channel capacity loss with CSI difference at the $(n,k)$-th subcarrier is given by
\begin{equation*}\label{Equ:CapacityLoss}
\mathscr{L}_{n,n_c}^{(k,k_c)} \leq {\log}_2 \det\left(\mathbf{I} + \gamma \mathbf{S}_{n}^{(k)}\right) - {\log}_2\det\left(\mathbf{I} + \gamma _e \mathbf{S}_{n}^{(k)}\right).
\end{equation*}
\end{lem}

It's worthy to note that the lower-bound of the channel capacity with CSI difference is  very important , since it can clearly characterize the worst case on channel capacity loss. In this way, we can derive the maximal capacity loss of channel with CSI difference at $(n,k)$-th subcarrier with \emph{Lemma} \ref{Lem:CapacityLossUpperbound} and  make guarantee the reliability  of our considered systems.

\subsection{Approximate Ergodic Capacity}

Notice that $\mathscr{L}_{n,n_c}^{(k,k_c)}$ varies with $\mathbf{S}_{n}^{(k)}$ and one cannot acquire $\mathbf{H}_{n_c}^{(k_c)}$ without the exact $(n_c,k_c)$, which are related to the subcarrier group size. Thus, it is reasonable to determine the size of subcarrier group first. In order to determine the group size, we need to find a capacity loss evaluation metric without the instantaneous of channel fading matrix $\mathbf{H}_{n_c}^{(k_c)}$. To this end, one needs to calculate the ergodic capacity loss and investigate the average performance of the subcarrier grouping system.

\begin{lem}\label{Lem:ErgodicCapacityLoss}
Define the center subcarrier as the $(n_c,k_c)$-th subcarrier. The ergodic capacity loss at the $(n,k)$-th subcarrier can be upper-bounded by

\begin{equation}\label{Equ:ErgodicCapacityLoss}
\begin{split}
&\mathscr{E}_{n,n_c}^{(k,k_c)} = \log_2\left( 1 + \sum_{k=1}^{v} \alpha_{k}(\gamma) \prod_{i=1}^{k-1} (u-i) \right) -\\
&{\qquad \qquad \qquad}  \log_2\left( 1 + \sum_{k=1}^{v} \alpha_{k}(\gamma_e) e^{\begin{matrix} \sum_{i=0}^{k-1} \psi (u-i) \end{matrix}} \right),
\end{split}
\end{equation}
where $\alpha _k(\gamma) = \gamma ^k v ( v - 1 ) {\cdot \cdot \cdot} (v - k + 1) / k!$. $\gamma$ and $\gamma_e$ is the $\rm{SNR}$ and $\rm{ESNR}$ at the $(n_c,k_c)$- and $(n,k)$-th subcarrier, respectively. $\psi(x)$ is defined by $\psi(x) = -\zeta + \sum_{r=1}^{x-1} \frac {1}{r}$, where $\zeta = 0.577215649{\cdot \cdot \cdot}$ is the Euler's constant. Moreover, the relative ergodic capacity loss is

\begin{equation}\label{Equ:RelativeErgodicCapacityLoss}
\begin{split}
&\mathscr{R}_{n,n_c}^{(k,k_c)} = 1- \frac{\log_2\left( 1 + \sum_{k=1}^{v} \alpha_{k}(\gamma_e) e^{\begin{matrix} \sum_{i=0}^{k-1} \psi (u-i) \end{matrix}} \right)}{\log_2\left( 1 + \sum_{k=1}^{v} \alpha_{k}(\gamma) \prod_{i=1}^{k-1} (u-i) \right)}.
\end{split}
\end{equation}

\end{lem}

\begin{proof}
See appendix \ref{Appen:ergodicCapacityThm}.
\end{proof}

\emph{Lemma} \ref{Lem:ErgodicCapacityLoss} shows that the ergodic capacity loss gives the average evaluation of the $(n,k)$-th subcarrier capacity loss. Hence, given a specific capacity loss threshold, one can easily derive the CSI difference parameter $\sigma_{m_{n,n_c}^{(k,k_c)}}^2$ with \eqref{Equ:SNROfCapacityBounds} and \eqref{Equ:ErgodicCapacityLoss}.

\section{Subcarrier Grouping and Performance Evaluation}\label{Sec:Algorithm}

According to the theoretical results obtained previously, in this section we shall employ the channel self-correlation property to propose an adaptive subcarrier grouping based precoding algorithm. It can be noted that the channel self-correlation is related to the transmit environment and mobile velocity. The sensing of transmit environment is out of the scope of this paper and the mobile velocity can be easily acquired via positioning systems, e.g. GPS, Beidou, etc. At last of this section, we analyze the performance of the proposed algorithm regarding the system cost, SNR and antenna number.

\subsection{Subcarrier Grouping Method in Double-Selective Channels}

Theoretical results in \emph{Lemma} \ref{Lem:ErgodicCapacityLoss} shows that the capacity loss are mainly depended on the second order statistics of $\mathbf{M}^{(k,k_c)}_{n,n_c}$. In order to investigate the relation between capacity loss and subcarrier group size, one needs to observe $\mathbf{M}^{(k,k_c)}_{n,n_c}$ associated with the transmit environment and mobile velocity, i.e. delay spread and Doppler spread. Its second moment can be characterized by channel self-correlation (denoted by $\mathbf{R}(f^{(k)},f^{(k_c)},t^{(n)},f^{(n_c)})$). That is
\begin{equation}\label{Equ:ChannelCorrelation}
\begin{split}
&\sigma_{\mathbf{M}^{(k,k_c)}_{n,n_c}} ^2 = 2 \bigg\{ \mathbf{R}\left(f^{(k_c)},f^{(k_c)},t^{(n_c)},t^{(n_c)}\right) \\
&{\qquad \qquad \qquad}  -\mathbf{R}\left(f^{(k)},f^{(k_c)},t^{(n)},t^{(n_c)}\right) \bigg\}.
\end{split}
\end{equation}
where $f^{(k)}$ and $f^{(k_c)}$ denote the carrier frequency of the $k$-th and $k_c$-th subcarrier, respectively. $t^{(n)}$ and $t^{(n_c)}$ denote the time slot of the $n$-th and $n_c$-th OFDM block, respectively. 

Note that the correlation function in \eqref{Equ:ChannelCorrelation} depends on four variables, with a rather complicated form.  For characterizing the CSI difference properties with the correlation function,  one usually introduces further assumptions about the physics of channel and simplify $\mathbf{R}(f^{(k)},f^{(k_c)},t^{(n)},f^{(n_c)})$. The most frequently used assumptions are \emph{Wide-Sense Stationary} (WSS) assumption and \emph{Uncorrelated Scatterers} (US) assumption, where WSS means that self-correlation depends only on time difference $t^{(n,n_c)} = |t^{(n)} - t^{(n_c)}|$, while US means that self-correlation depends only on frequency difference $f^{(k,k_c)} = |f^{(k)}-f^{(k_c)}|$. By involving both assumptions simultaneously as \emph{WSSUS} model \cite{molisch2007wireless}, we can simplify \eqref{Equ:ChannelCorrelation} as

\begin{equation}\label{Equ:ChannelCorrelationWSSUS}
\sigma _{\mathbf{M}^{(k,k_c)}_{n,n_c}} ^2 = 2 \left( \mathbf{R}\left(0,0\right) -\mathbf{R}\left(f^{(k,k_c)},t^{(n,n_c)}\right) \right).
\end{equation}

Furthermore, using the independence between the delay spread and frequency spread of double-selective channels, one can decompose the self-correlation of channel response into as 

\begin{equation}\label{Equ:CorrelationDecomposition}
\begin{split}
\mathbf{R}\left(f^{(k,k_c)},t^{(n,n_c)}\right) &= \mathbf{R}\left(f^{(k,k_c)},0\right)  \mathbf{R}\left(0,t^{(n,n_c)}\right) \\
&= \mathbf{R}_f\left(f^{(k,k_c)}\right)  \mathbf{R}_t\left(t^{(n,n_c)}\right).
\end{split}
\end{equation}
where $\mathbf{R}_f\left(f^{(k,k_c)}\right)$ and $\mathbf{R}_t\left(t^{(n,n_c)}\right)$ are defined as
\begin{equation}\label{Equ:RtRf}
\begin{cases}
\mathbf{R}_f\left(f^{(k,k_c)}\right) \triangleq  \mathbf{R}\left(f^{(k,k_c)},0\right)\\
\mathbf{R}_t\left(t^{(n,n_c)}\right) \triangleq  \mathbf{R}\left(0,t^{(n,n_c)}\right),
\end{cases}
\end{equation}
 both of them denote the frequency-domain correlation between the $k$- and $k_c$-th subcarrier in one OFDM block and the time-domain correlation between the $n$- and $n_c$-th OFDM block at one subcarrier, respectively.

As $\mathbf{H}_{n,n_c}^{(k,k_c)}$ is normalized and its entries are assumed to be i.i.d. Rayleigh fading, the correlation matrices $\mathbf{R}\left(0,0\right)$ and $\mathbf{R}\left(f^{(k,k_c)},t^{(n,n_c)}\right)$ are $v\text{-}by\text{-}v$ diagonal matrices with equal diagonal elements $v$ and $vR\left(f^{(k,k_c)},t^{(n,n_c)}\right)$. In addition, $R\left(f^{(k,k_c)},t^{(n,n_c)}\right) = \mathrm{E}\left(h^{(k)}_{n} h^{(k_c)}_{n_c}\right)$, where $h^{(k)}_{n}$ and $h^{(k_c)}_{n_c}$ are the entries of $\mathbf{H}^{(k)}_{n}$ and $\mathbf{H}^{(k_c)}_{n_c}$. Thus, \eqref{Equ:ChannelCorrelationWSSUS} can be simplified as
\begin{equation}\label{Equ:ChannelCorrelationOfH}
\sigma _{\mathbf{M}^{(k,c)}_{n,n_c}}^2 = v\sigma _{{m}^{(k,k_c)}_{n,n_c}} ^2 \mathbf{I}_v 
\end{equation}
where 
\begin{equation}\label{Equ:ChannelCorrelationOfh}
\sigma _{m^{(k,k_c)}_{n,n_c}}  ^2 = 2\left(1 - R_f\left(f^{(k,k_c)}\right)R_t\left(t^{(n,n_c)}\right)\right)
\end{equation}

Since the channel fading matrix is normalized, i.e. $\rm{trace}\left(\mathbf{S}_n\right) / v = 1$, it can be inferred that $\sigma _{{m}^{(k,k_c)}_{n,n_c}}  ^2  \in [0,2]$.

Substituting \eqref{Equ:SNROfCapacityBounds} and \eqref{Equ:ChannelCorrelationOfh} into \eqref{Equ:RelativeErgodicCapacityLoss} and denote $\mathscr{R}_{n,n_c}^{(k,k_c)}$ as $\mathscr{R}\left(R_f,R_t\right)$, one can observe that $R_f\left(f^{(k,k_c)}\right)$ and $R_t\left(t^{(n.n_c)}\right)$ play the similar role in $\mathscr{R}\left(R_f,R_t\right)$. Thus, for a given constraint of the ergodic capacity loss $\mathscr{R}\left(R_f,R_t\right) \leq \zeta_r$, the feasible zone is showed in Fig. \ref{Fig:RdomainCloss}. It is observed that when denote (1,1) as the origin, the distance from the origin to the intersections of the border and the ordinate ( abscissa ) is  $R_f^{-1}(\zeta_r)$ ( $R_t^{-1}(\zeta_r)$ ). The feasible zones (i.e. the areas inside the curves) is convex, which is shown in the following lemma.

\begin{figure}[htbp]
\centering
\includegraphics[width=0.48\textwidth]{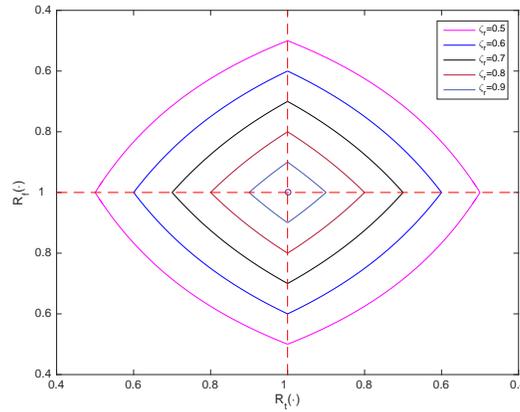}
\caption{The feasible zone of correlations when relative capacity loss threshold $\zeta_r$=0.5, 0.6, 0.7, 0.8 and 0.9.} \label{Fig:RdomainCloss}
\end{figure}

\begin{lem}\label{lem:ConvexMap}
The feasible zone described above is convex, i.e.

\begin{equation*}
\begin{split}
\lambda \mathscr{R}(R_{f1},R_{t1}) &+ (1-\lambda) \mathscr{R}(R_{f2},R_{t2}) \geq \\
& \mathscr{R}(\lambda R_{f1} + (1-\lambda) R_{f2}, \lambda R_{t1} + (1-\lambda) R_{t2}),\\
\end{split}
\end{equation*}
where $\lambda \in $ [0,1]. $R_{f1}$ and $R_{f2}$ as well as $R_{t1}$ and $R_{t2}$ are the instance of frequency and time domain correlations, respectively. Furthermore,

\begin{equation*}
\begin{split}
\lambda \mathscr{R}&(R_f(f_1), R_t(t_1)) + (1-\lambda) \mathscr{R}(R_f(f_2), R_t(t_2)) \geq \\
& \mathscr{R}(R_f( \lambda f_1+ (1-\lambda) f_2 ), R_t(\lambda t_1 + (1-\lambda) t_2))\\
\end{split}
\end{equation*}
holds. $f_1$ and $f_2$ as well as $t_1$ and $t_2$ are the instance of frequency and time, respectively.
\end{lem}

\begin{proof}
Notice that $\mathscr{R}(R_f,R_t) = \zeta_r$ is equivalent to $\sigma _{m^{(k,k_c)}_{n,n_c}}  ^2 = 2(1 - R_f(f^{(k,k_c)})R_t(t^{(n,n_c)})) = C$, where $C$ is a constant derived from \eqref{Equ:RelativeErgodicCapacityLoss}. Thus, the curve is a hyperbola and the convex property is obvious. Considering the correlation function $R_f(f)$ and $R_t(t)$ are monotone decreasing, the convex property holds when using $t$ and $f$ as the abscissa and ordinate.
\end{proof}

The second equation in \emph{Lemma} \ref{lem:ConvexMap} shows that the convex property holds for the boundary, when using $t$ and $f$ as the abscissa and ordinate. Thus, connect the four intersections between the boundary and coordinate axes with segments and form a rhombus with the transverse and longitudinal diagonal length as $2R_t^{-1}(\zeta_r)$ and $2R_f^{-1}(\zeta_r)$, all the subcarriers inside this rhombus satisfy the capacity loss constraint $\mathscr{R}(R_f,R_t) \leq \zeta_r$. Thus, the subcarriers inside one rhombus can be grouped and apply a unitary precoding matrix, which is shown in Fig. \ref{Fig:CellularNetwork}. It can be noticed that the subcarrier grouping scheme varies with channel correlation status, which is determined by transmit environment and mobile velocity. The theoretical result is summarized in  Theorem 1.

\begin{figure}
\centering
\includegraphics[width=0.45\textwidth]{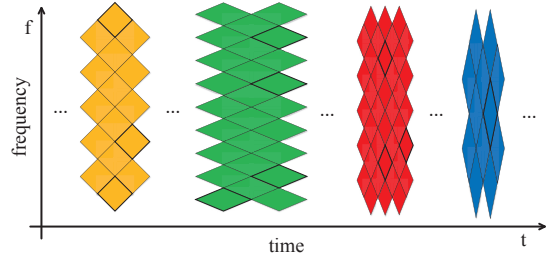}
\caption{Grouping scheme vaires with the channel correlation status.} \label{Fig:CellularNetwork}
\end{figure}

\begin{thm}\label{Thm:absConstraint}
Given the capacity loss constraint $\mathscr{R}(R_f,R_t) \leq \zeta_r$  with transmit power $\rho_t$, all the subcarriers with identification $(n,k)$ that satisfy
\begin{equation*}
\pounds (n_c,k_c) = \left\{(n,k)|R(f^{(k,k_c)},  t^{n,n_c}) \geq \frac{2-\sigma _{m^{(k,k_c)}_{n,n_c}}  ^2}{2}\right\}
\end{equation*}
are grouped into subcarrier group $\pounds(n_c,k_c)$. The second moment $\sigma _{m^{(k,k_c)}_{n,n_c}}  ^2$ is given by

\begin{equation*}
\begin{split}
1- \frac{\log_2\left( 1 + \sum_{k=1}^{v} \alpha_{k}(\gamma_e) e^{\begin{matrix} \sum_{i=0}^{k-1} \psi (u-i) \end{matrix}} \right)}{\log_2\left( 1 + \sum_{k=1}^{v} \alpha_{k}(\gamma) \prod_{i=1}^{k-1} (u-i) \right)} = \zeta_r,
\end{split}
\end{equation*}
where
\begin{equation*}
\left\{
\begin{aligned}
\gamma _e &= \frac{\sigma _{x_{n}^{(k)}} ^2}{v \sigma_{m_{n,n_c}^{(k,k_c)}} ^2 \sigma _{x_{n}^{(k)}} ^2 + \sigma_{w_{n}^{(k)}}^2} = \frac{\rho_t}{v \sigma_{m_{n,n_c}^{(k,k_c)}} ^2 \rho_t + N_t} \\
\gamma &= \frac{\sigma _{x_{n}^{(k)}} ^2}{\sigma_{w_{n}^{(k)}}^2} = \frac{\rho_t}{N_t}.
\end{aligned}
\right\}.
\end{equation*}

where $\beta = \frac{2-\sigma _{m^{(k,k_c)}_{n,n_c}}  ^2}{2}$ represents the correlation threshold. The subcarrier grouping method is to divide the whole space with a serious of rhombuses, where the transverse and longitudinal diagonal is $2R_t^{-1}(\beta)$ and $2R_f^{-1}(\beta)$, respectively. All the subcarriers inside the same rhombus can be put into one group and apply the same precoding matrix.
\end{thm}

 Theorem 1 indicates that the grouping scheme is mainly determined by the capacity loss threshold $\zeta_r$ and channel correlation status $R(f,t)$ in double selective fading channels. Hence, the grouping scheme is stable when capacity loss threshold $\zeta_r$ and wireless channel correlation status remain unchanged. On the contrary, if channel correlation status varies, the corresponding  grouping scheme must be modified immediately for sustaining channel capacity loss threshold, which is shown in Fig. \ref{Fig:CellularNetwork}. Algorithm \ref{Alg:Framwork} explains the corresponding subcarrier grouping procedure in detail.

Furthermore, we also derive the relation among the subcarrier group size, $\zeta_r$ and $R(f,t)$, which is given in the following theorem.

\begin{thm}\label{Thm:SubcarrierGroupingSize} In double-selective channels, given the channel correlation threshold $\beta$, which can be driven from theorem \ref{Thm:absConstraint}, the subcarrier group size can be determined by

\begin{equation}\label{Equ:calculateGroupSize}
S = \frac{1}{2}\lfloor\frac{2R_f^{-1}(\beta)}{B_f}\rfloor \lfloor\frac{2R_t^{-1}(\beta)}{B_t}\rfloor,
\end{equation}
where the parameter $B_f$ and $B_t$ are the frequency distance between adjacent subcarriers and the time interval between adjacent OFDM blocks, respectively.
$\lfloor x \rfloor$ denotes the largest integer less than or equal to $x$.
\end{thm}

\begin{proof}
For the subcarriers are grouped into a serious of rhombuses, where the transverse and longitudinal diagonal length equal to $2R_t^{-1}(\beta)$ and $2R_f^{-1}(\beta)$, respectively. The corresponding OFDM block number and subcarrier number are, respectively,
\begin{equation*}
S_f = \lfloor\frac{2R_f^{-1}(\beta)}{B_f}\rfloor
\end{equation*}
and
\begin{equation*}
S_t = \lfloor\frac{2R_t^{-1}(\beta)}{B_t}\rfloor.
\end{equation*}
Therefore, the number of subcarriers contained in this subcarrier grouping rhombus can be driven by the area formula of rhombuses, i.e.
\begin{equation*}
S = \frac{1}{2} S_t S_f = \frac{1}{2}\lfloor\frac{2R_f^{-1}(\beta)}{B_f}\rfloor \lfloor\frac{2R_t^{-1}(\beta)}{B_t}\rfloor.
\end{equation*}
\end{proof}

This theorem can be used to estimate algorithm complexity and computation cost for subcarrier grouping.

\begin{algorithm}[htb] 
\renewcommand{\algorithmicrequire}{\textbf{Input:}}
\renewcommand\algorithmicensure {\textbf{Output:}}
\renewcommand\algorithmicrepeat {\textbf{Repeat:}}
\renewcommand\algorithmicuntil {\textbf{Until:}}
\caption{Precode with subcarrier grouping method.} 
\label{Alg:Framwork} 
\begin{algorithmic}
\Require\\
\begin{enumerate}
\item The set of system parameters, $\{N_t , N_r , K\}$;  
\item Channel state training sequence, $S$;  
\item The working SNR at the receivers, $\gamma = \sigma _x ^2 / \sigma _w ^2$;  
\item Capacity loss threshold, $\zeta_r$.
\end{enumerate}

\Ensure precoded signals $\mathbf{F}_{n_c}^{(k_c)} \mathbf{X}_{n}^{(k)}$\\ 
\textbf{Initialization:} Estimate the channel self-correlation $R(t,f)$ and give the subcarrier grouping scheme by \emph{Theorem} \ref{Thm:absConstraint} with capacity loss threshold $\zeta_r$; \label{Alg:CalculateN}

\Repeat\\
\begin{enumerate}
\item Estimate the channel state information $\mathbf{H}_{n_c}^{(k_c)}$ at the center of each subcarrier groups from the training sequence $S$; \label{Alg:EstimateCSIT}

\item Generate $\mathbf{F}_{n_c}^{(k_c)}$ with $\mathbf{H}_{n_c}^{(k_c)}$; \label{Alg:GeneratePrecoder}

\item Get the CSI difference $\sigma_{m_{n,n_c}^{(k,k_c)} }^2$ by \eqref{Equ:ChannelCorrelationOfh} and derive the exact capacity lower bounds at each individual subcarriers by \eqref{Equ:BoundsOfMismatch}; \label{Alg:ExactCapacity}

\item Arrange the bit rates to symbols at each individual subcarriers by the capacities derived in Step. \ref{Alg:ExactCapacity} and precode the symbols with $\mathbf{F}_{n_c}^{(k_c)}$;
\end{enumerate}

\Until{The subcarrier grouping size changed.}

\end{algorithmic}
\end{algorithm}

\subsection{Performance Analysis}\label{Sec:PerformanceAnalysis}

In this subsection, we give the brief performance analysis of proposed subcarrier grouping scheme regarding the system cost, SNR and antenna number.

\subsubsection{System Cost}
The system cost in one single subcarrier is considered. Comparing to the subcarrier grouping scheme with fixed group size, the proposed scheme needs to estimate channel correlation status. Denote the subcarrier number in the considered domain is $K M_{t}$, where $K$ is subcarrier number in one OFDM block and $M_{t}$ is the number of OFDM blocks. As known to all, the most common precoding algorithm in each individual subcarrier requires $O(v^3)$ operations\cite{vu2007mimo,jiang2005joint,jiang2005uniform}. Thus, the precoding complexity on all subcarriers is $O(KM_{t} v^3)$. In addition, estimating channel correlation requires $\left(O(K^2 M_{t})+O(M_t^2)\right)$ operations\cite{therrien1992discrete}.

Note that the computation of correlation threshold $\beta$ only needs to calculate a deterministic function one time in the considered domain. Hence, its computation complexity is $O(1)$. Thus, the total system computational complexity composed by the computing of precoding matrices and the estimation of channel correlation is
\begin{equation}\label{Equ:complexity}
\begin{split}
\digamma &= O \left( \frac{KM_{t}}{S} v^3 \right) + O(K^2 M_{t})+O(M_t^2)+O(1)\\
&= O \left( \frac{KM_{t}}{S} v^3 + K^2 M_{t} +M_t^2 \right),
\end{split}
\end{equation}
where $S$ is the subcarrier group size denoted in \emph{Theorem} 2 . Moreover, the average precoding complexity on each individual subcarriers in  the considered domain is
\begin{equation}\label{Equ:Avecomplexity}
\begin{split}
\digamma_{ave} = \frac{\digamma}{K M_t} &= O \left( \frac{v^3}{S} \right) + O(K)+O \left( \frac{M_t}{K} \right)\\
&= O \left( \frac{v^3}{S}  + K+ \frac{M_t}{K} \right).
\end{split}
\end{equation}

\subsubsection{Signal to Noise Ratio}
\emph{Lemma} \ref{Lem:ErgodicCapacityLoss} shows that the introduced interference caused by CSI difference can reduce the channel capacity, especially in the scenarios with high SNR. The effects of SNR on capacity loss is analyzed as follows in the scenarios with low and high SNRs, respectively.

\emph{a) Low \rm{SNR} scenarios}:

The ESNR satisfies $\lim\limits_{\gamma \rightarrow{0}} \gamma_e = \gamma$. Hence, $\lim\limits_{\gamma \rightarrow{0}} C_{n}^{(k)} = C_{n_c}^{(k_c)}$ and the capacity loss $\mathscr{L}_{n,n_c}^{(k,k_c)}$ can be neglected in low SNR scenarios.

\emph{b) High \rm{SNR} scenarios}:

The ESNR satisfies
$\lim\limits_{\gamma \rightarrow{\infty}} \gamma_e = \frac{1}{v \sigma_{m_{n,n_c}^{(k,k_c)}}^{2}}$ and $\lim\limits_{\gamma \rightarrow{\infty}} C_{n}^{(k)} = \log_{2} \det \left( \mathbf{I} + \frac{\mathbf{S}_n^{(k)}}{v\sigma_{m_{n,n_c}^{(k,k_c)}} ^2} \right)$, which is a constant. Thus, channel capacity at the $(n,k)$-th subcarrier cannot be improved limitlessly by increasing SNR individually, and capacity loss $\mathscr{L}_{n,n_c}^{(k,k_c)}$ approaches infinity at extremely high SNR regimes. Therefore, properly SNR needs to carefully generated to achieve the balance between channel capacity and capacity loss at the $(n,k)$-th subcarrier. Simulation results in Section \ref{sec:SimulationSNR} shows that the optimum SNR is

\begin{equation}\label{equ:optimumSNR}
\gamma = \frac{1}{v \sigma_m^2}.
\end{equation}
In that SNR, the introduced interference is equal to the AWGN and system reaches the balance between spectrum efficiency and energy efficiency.


\subsubsection{antenna number}
The communication system in this paper is assumed to be power constrained. If the transmit power is uniformly allocated to all transmit antennas, then the SNR and ESNR at the $(n_c,k_c)$- and $(n,k)$-th subcarrier are
\begin{equation*}
\gamma = \frac{\rho}{N_t}
\end{equation*}
and 
\begin{equation*}
\gamma _ {e}^{(n,k)} = \frac{\rho/N_t}{v \sigma_{m_{n,n_c}^{(k,k_c)}} ^2 \rho/N_t+1},
\end{equation*}
respectively.

The property of $\mathscr{R}_{n,n_c}^{(k,k_c)}$ with respect to antenna number can be summarized by following Corollary.

\begin{cor}\label{cor:antennaNumber}
The relative capacity loss is monotonic decreasing with respect to the antenna number. As the antenna number approaches infinity,  $\mathscr{R}_{n,n_c}^{(k,k_c)}$ converges to
\begin{equation}\label{equ:AntennaLimitCapacityloss}
\lim\limits_{v \rightarrow \infty} \mathscr{R}_{n,n_c}^{(k,k_c)} = 1 - \frac{ \log\left( 1+\frac{\rho_t}{\rho_t \sigma _{{m}^{(k,k_c)}_{n,n_c}}  ^2  + 1 } \right) }{\log(1+\rho_t)}.
\end{equation}
where the center subcarrier of this subcarrier group is assumed to be $(n_c,k_c)$
\end{cor}

\begin{proof}
The monotonicity of $\mathscr{R}_{n,n_c}^{(k,k_c)}$ is easily to get.  Since the total power is fixed, the power allocated into each single transmit antenna will decreases as the antenna number increases, resulting in $\mathscr{R}_{n,n_c}^{(k,k_c)}$ decreasing with respect to $\gamma$.

As the antenna number goes to infinity, $v \rightarrow \infty$, e.g. in the massive MIMO regimes, $\lim\limits_{v \rightarrow \infty} \mathbf{S}_{n}^{(k)} = u \mathbf{I}_v$\cite{sellathurai2010mimo}. In this case, the channel capacity at the $(n,k)$-th subcarrier is
\begin{equation*}
C_{n}^{(k)} = v \log \left(1+\frac{\rho_t}{\rho_t \sigma _{{m}^{(k,k_c)}_{n,n_c}}  ^2  + 1} \right),
\end{equation*}
and the channel capacity at the $(n_c,k_c)$-th subcarrier is
\begin{equation*}
C_{n_c}^{(k_c)} = v\log(1+\rho_t).
\end{equation*}

Hence, the relative capacity loss at the $(n,k)$-th subcarrier is
\begin{equation*}
\mathscr{R}_{n,n_c}^{(k,k_c)} = \frac{ C_{n_c}^{(k_c)} - C_{n}^{k} }{C_{n_c}^{(k_c)}} = 1 - \frac{ \log\left( 1+\frac{\rho_t}{\rho_t \sigma _{m_{n,n_c}^{(k,k_c)}}^2  + 1 } \right) }{\log(1+\rho_t)} ,
\end{equation*}
which is constant.
\end{proof}

\section{Numerical Results}\label{Sec:NumRes}

In this section, some numerical results are presented to show the validity of our theoretical analysis and the effectiveness of proposed adaptive subcarrier grouping based precoding algorithm.  In the simulations, the total transmit power $\rho$ is allocated to each transmit antenna equally. Hence, the SNR at each individual antennas is $\gamma = \rho / N_t$. $B_f$ is assumed as 14KHz and $B_t = \frac{1}{B_f} = 71 us$, referencing datas from \cite{dahlman20073g}. The fading channels are assumed to be i.i.d. Rayleigh channels.

With these assumptions, we evaluate the theoretical results with COST207 channel model and Jakes multipath model, which are widely used for simulating the frequency-selectivity and the time-selectivity of wireless channel fading models. For the monotonicity of the channel capacity loss function with respect to SNR and the variance of CSI difference, the variation tendency between absolute capacity loss and relative capacity loss is identical. Thus, for sake of concise, only the results on relative capacity loss is simulated.

\subsection{Theoretical Results on variant correlation thresholds and antenna numbers}\label{sec:SimulationSNR}

The capacity loss on variant channel correlation is depicted in Fig. \ref{Fig:GroupingRelativeCapacityLoss}, wherein the antenna regime is $8 \times 8$. It can be observed that the relative capacity loss grows in S-type regarding SNR. This is because in low, medium and high SNR regimes, $v\sigma_{m_{n,n_c}^{(k,k_c)}} ^2 \sigma _{x_{n}^{(k)}} ^2$ (i.e. introduced interference) is relatively low, comparable and much too high compared to $\sigma_{w_{n}^{(k)}}^2$ (i.e. AWGN), respectively. Therefore, the relative capacity loss grows sharply at medium SNR, and as shown in previous analysis in Section \ref{Sec:PerformanceAnalysis},  the relative capacity loss grows slowly at low and high SNR. Thus, proposed subcarrier grouping scheme may perform well in low SNR regimes and achieve balance between channel capacity and capacity loss in medium SNR, which has been mentioned in \eqref{equ:optimumSNR}.

Another important observation is that with the increasing of $\beta$ (i.e. correlation threshold), the system can acquire less capacity loss at the same SNR, which is intuitionistic. This is because when $R_{tf}$ (i.e. the channel correlation between the $(n,k)$- and $(n_c,k_c)$-th subcarrier) increases, the required SNR must be lower in order to keep $v\sigma_{m_{n,n_c}^{(k,k_c)}} ^2 \sigma _{x_{n}^{(k)}} ^2$ invariable.

The effects of antenna numbers are shown in Fig. \ref{Fig:AntennaNumber}, in which $\beta$=0.95 and SNR=10dB, respectively. It can be observed that $\mathscr{R}_{n,n_c}^{(k,k_c)}$ decreases with the growth of antenna numbers. Especially, it reduces sharply at the first several antenna numbers. Thus, properly antenna numbers can efficiently reduce the capacity loss caused by CSI difference. Another important observation is that, as antenna number further increases, the relative capacity loss converges to 25.28\%. Hence, one cannot decrease the capacity loss limitlessly by increasing antenna number, because the capacity loss is lower-bounded by \eqref{equ:AntennaLimitCapacityloss}.

\begin{figure}
\centering
\includegraphics[width=0.45\textwidth]{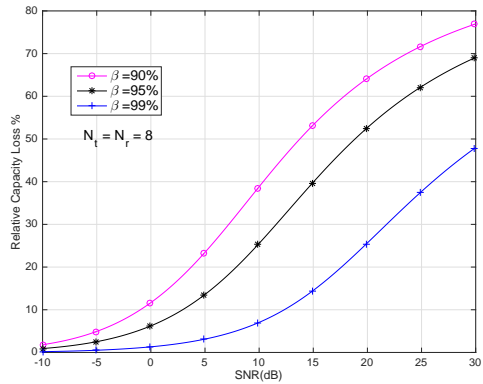}
\caption{Relative capacity loss versus correlation $\sigma_{M}$=0.01, 0.05, 0.1 when MIMO antenna number $N_t=N_r=8$.} \label{Fig:GroupingRelativeCapacityLoss}
\end{figure}

\begin{figure}
\centering
\includegraphics[width=0.45\textwidth]{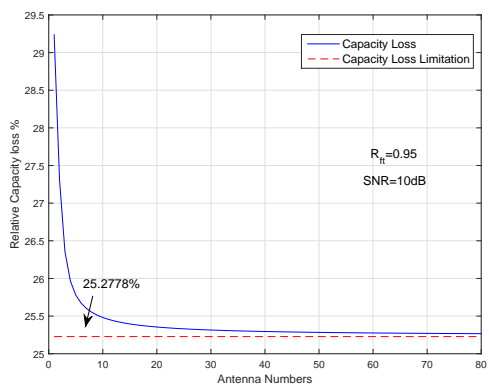}
\caption{Relative capacity loss versus MIMO antenna number $N_t=N_r={1,2,\cdots,80}$ and the limitation when $min\{N_t,N_r\} \rightarrow \infty$ when constraint $R_{tf}=0.95$ and SNR = 10dB.} \label{Fig:AntennaNumber}
\end{figure}

\subsection{Grouping Performance of Proposed Algorithm}

For evaluating the algorithm performance and summarizing the application conditions, we build different transmission scenarios with the help of COST207 and Jakes model, which are widely used in wireless channel simulations. The subcarrier grouping performance at different SNR, moving velocity and transmit environment are analyzed numerically in this part.

The effects on SNR of 0 dB, 10 dB and 15 dB are shown in Fig. \ref{Fig:SubcarrierGroupingSize_SNR}, where the transmitting environment is RA (i.e. rural area) and the moving velocity is 10 m/s. It can be observed that the subcarrier group sizes at the three SNRs grow exponentially with the capacity loss threshold, and in the 0 dB regime, the grouping size grows about 5 times faster than that in the 15 dB regime. It means that the system cost can be significantly reduced by sacrificing channel capacity and it's especially efficient in low SNR regime.

In addition, from \eqref{Equ:calculateGroupSize} one can observe that the grouping size is determined by the subcarrier grouping length in frequency domain as well as time domain (i.e. $S_f$ and $S_t$). As shown in Fig. \ref{Fig:SubcarrierGroupingSize_SNR}, the contributions from $S_f$ and $S_t$ are denoted as frequency stages and time stages, respectively. It can be noted that, in simulated transmitting regimes, time stages is relatively smaller than the frequency stages. Thus, the marked time stages in Fig. \ref{Fig:SubcarrierGroupingSize_SNR} are the time stages with one unitary $S_f$. The stages are caused by the varying of $S_t$.

The frequency stages are related to channel correlations, i.e. multipath scattering and small scale fading. As to the time stages, they are related to channel time variations, i.e. Doppler spread. Therefore, the frequency stages and time stages are separately determined by transmitting environments and antenna moving velocities, which are depicted in  Fig. \ref{Fig:SubcarrierGroupingSize_ENV}-\ref{Fig:SubcarrierGroupingSize_Speed}.

\begin{figure}
\centering
\includegraphics[width=0.45\textwidth]{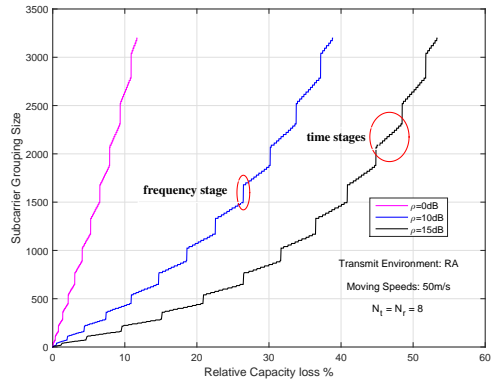}
\caption{Subcarrier grouping performance versus variant SNR (0 dB, 10 dB and 15 dB) when the transmit environment is rural area (RA), moving speed is 10 m/s and the MIMO antenna regime is $8 \times 8$.} \label{Fig:SubcarrierGroupingSize_SNR}
\end{figure}

Fig. \ref{Fig:SubcarrierGroupingSize_ENV} shows that in rich scattering environments (e.g. typical urban (TU)), the subcarrier group size suffers 50\% loss and 80\% more loss compared to low scattering environments (e.g. hilly terrain (HT) and RA). Observe the frequency stages in three simulated environments, one can observe that the frequency stages is much smaller in rich scattering transmitting scenarios. This is because as the channel scattering increases, the channel frequency correlations decreases. Thus, smaller frequency subcarrier grouping length is needed to achieve the same capacity loss, which leads to smaller frequency stages.

Fig. \ref{Fig:SubcarrierGroupingSize_Speed} shows that compared to 10 m/s and 30 m/s mobility scenarios, the subcarrier group size in 100 m/s suffers 70\% and 40\% more capacity loss, respectively. It can be observed that the length of frequency stages in this three simulated velocities are equal. This is because the transmit environment is fixed and thus, the channel frequency correlations in these three simulation regimes are identical. As to the time stages, that in low velocity scenarios is obviously larger than the one in high velocity scenarios. That is, in low velocity scenarios, the channel time variation is slow and larger time subcarrier grouping length can be achieved.

From the simulation results shown in Fig. \ref{Fig:SubcarrierGroupingSize_ENV}-\ref{Fig:SubcarrierGroupingSize_Speed}, we can simply justify that the SNR, environmental multipath scattering and Doppler frequency shift can greatly affect the subcarrier grouping performance. We can also get some insights on the suitable conditions of proposed subcarrier grouping scheme: the subcarrier group size can be extremely large in some specific scenarios (i.e. sparse scattering environment or low Doppler frequency spread scenarios). For instance, rural areas, low mobility, etc. In these suitable scenarios, the proposed subcarrier can achieve low system cost or low capacity loss.

\begin{figure}
\centering
\includegraphics[width=0.45\textwidth]{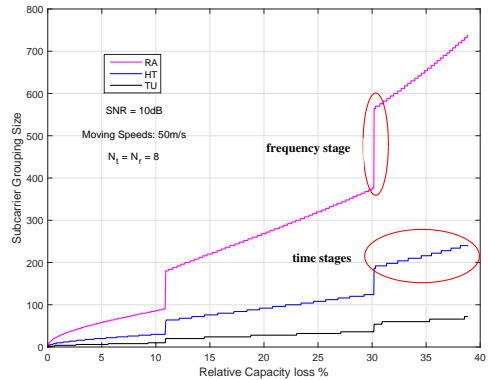}
\caption{Subcarrier grouping performance versus variant environments (RA, HT and TU) when the SNR is 10dB, moving speed is 10m/s and the MIMO antenna number is $N_t=N_r=8$.} \label{Fig:SubcarrierGroupingSize_ENV}
\end{figure}

\begin{figure}
\centering
\includegraphics[width=0.45\textwidth]{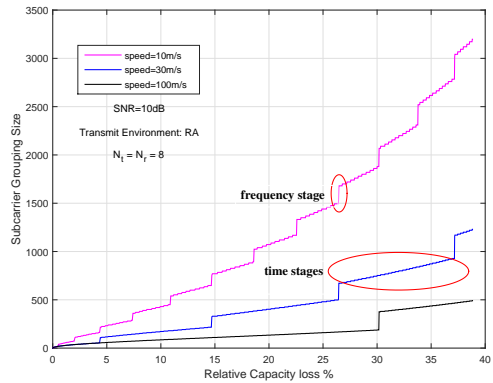}
\caption{Subcarrier grouping performance versus variant moving speeds (10m/s, 30m/s, 100m/s) when the SNR is 10dB, transmit environment is rural area (RA) and the MIMO antenna number is $N_t=N_r=8$.} \label{Fig:SubcarrierGroupingSize_Speed}
\end{figure}

Furthermore, we consider the other subcarriers in the same subcarrier group. As known to all, some continuous functions can be approximately denoted by the Taylor expansion of itself, if its derivatives exists. For our discussed subcarrier grouping problem, in order to avoid excess capacity loss, the correlation threshold $\beta$ can not be too large. Thus, either the channel correlation function $R_{tf}(t,f)$ or the relation functions \eqref{Equ:RelativeErgodicCapacityLoss} and \eqref{Equ:ErgodicCapacityLoss} can be approximated by using its first order derivative, meaning that we can generate the CSI difference in linear mode. That is, the capacity loss of the entire system is the average capacity loss of the best subcarrier and the worst subcarrier. Therefore, we can summarize that the system capacity loss is half of the worst subcarrier capacity loss, since the best subcarrier is at the center of the subcarrier group and it's corresponding capacity loss is zero.

\section{Conclusion}\label{Sec:Conclusion}

In this paper, we analyzed the relations between the OFDM subcarrier group size and the corresponding system capacity loss over i.i.d. double-selective correlated Rayleigh fading MIMO channels. Under the bridge of the second order statistic of CSI difference between the $(n,k)$- and $(n_c,k_c)$-th subcarriers  (i.e. $\mathbf{M}_{n,n_c}^{(k,k_c)}$), the quantitative mapping from subcarrier group size to channel capacity loss is derived. It can be noted that the second order statistic of CSI difference is related to channel self-correlation and is determined by transmit environments and antenna velocity, which can be sensed via specific signal processing methods and positioning systems, respectively.

In summary, theoretical results in this paper show that the subcarrier group size is related to SNR, MIMO structure and channel self-correlation, when the capacity loss is constraint. More specifically, the capacity loss decreases with respect to antenna numbers and channel correlation. By contrast, it increases with respect to SNR. One more important observation is that the relative capacity loss decreases sharply with the first few antennas, and converges to a lower-bound expeditiously. Thus, when the total transmit power is constant, few additional antennas can efficiently decrease the capacity loss, but much too more additional antennas are not helpful.

Based on the developed theoretical results, we also provided an subcarrier grouping algorithm, which adaptively adjust the subcarrier group size under the sensed transmit environment to satisfy the preset capacity loss threshold. If one adjusts the capacity loss threshold, the proposed subcarrier grouping algorithm can achieve tradeoff between system service capacity and system cost. Simulation results show that proposed tradeoff can be extremely effective in some specific transmit regimes (i.e. low SNR, properly high antenna numbers or sparse scattering scenarios). In these transmit regimes, sacrificing acceptable system capacity, one can achieve large subcarrier group size and reduce system resource cost on subcarrier precoding efficiently, which agrees with developed theoretical results and provide valuable insights for the design of MIMO-OFDM communication systems.

\section*{Acknowledgment}

This work was partly supported by the China Major State Basic Research Development Program (973 Program) no.2012CB316100(2), National Natural Science Foundation of China (NSFC) no. 61171064 and NSFC no. 61321061.

\begin{appendices}

\section{}\label{Appen:capacityThm}
A similar range of the mutual information on the channel with CSI difference has been given in \cite{medard2000effect}, i.e.
\begin{equation}\label{Equ:MDLowerBound}
\begin{split}
& \mathbf{I} (\mathbf{X}_{n}^{(k)};\mathbf{Y}_{n}^{(k)} | \mathbf{H}_{n}^{(k)}) \geq  \\
&\log_2 \det \left( \mathbf{I} + \mathbf{H}_{n}^{(k)} \sigma _{\mathbf{X}_{n}^{(k)}} ^2  {\mathbf{H}_{n}^{(k)}}^\dagger ( \sigma_{\mathbf{M}_{n,n_c}^{(k,k_c)} \mathbf{X}_{n}^{(k)} }^2 + \sigma_{\mathbf{W}_{n}^{(k)}} ^2  )^{-1}   \right)
\end{split}
\end{equation}

The lower-bound of the channel capacity on the $(n,k)$-th subcarrier can be acquired via following derivations.

As to the independence between the variables $\mathbf{X}_{n}^{(k)}$, $\mathbf{H}_{n}^{(k)}$, $\mathbf{M}_{n,n_c}^{(k,k_c)}$ and $\mathbf{W}_{n}^{(k)}$, we have $\sigma_{\mathbf{M}_{n,n_c}^{(k,k_c)} \mathbf{X}_{n}^{(k)} }^2$ = $\sigma_{\mathbf{M}_{n,n_c}^{(k,k_c)}}^2 \sigma _{\mathbf{X}_{n}^{(k)}} ^2$ and the square matrix $\sigma _{\mathbf{X}_{n}^{(k)}} ^2$, $\sigma_{\mathbf{M}_{n,n_c}^{(k,k_c)}}^2$ and $\sigma_{\mathbf{W}_{n}^{(k)}} ^2$ are diagonal matrices with the same diagonal element. Thus, the right hand of (\ref{Equ:MDLowerBound}) can be rewritten as
\begin{equation}
\begin{split}
&\log  \det \left( \mathbf{I} + \mathbf{H}_{n}^{(k)} \sigma _{\mathbf{X}_{n}^{(k)}} ^2  {\mathbf{H}_{n}^{(k)}}^\dagger ( \sigma_{\mathbf{M}_{n,n_c}^{(k,k_c)} \mathbf{X}_{n}^{(k)} }^2 + \sigma_{\mathbf{W}_{n}^{(k)}} ^2  )^{-1}   \right) \\
= &\log  \det \left( \mathbf{I} + \sigma _{\mathbf{X}_{n}^{(k)}} ^2 \mathbf{S}_{n}^{(k)}  ( \sigma_{\mathbf{M}_{n,n_c}^{(k,k_c)} \mathbf{X}_{n}^{(k)} }^2 + \sigma_{\mathbf{W}_{n}^{(k)}} ^2  )^{-1}   \right) \\
= &{\log}_2~\det(\mathbf{I} + \frac{\sigma _{x_{n}^{(k)}} ^2}{v\sigma_{m_{n,n_c}^{(k,k_c)}} ^2 \sigma _{x_{n}^{(k)}} ^2 + \sigma_{w_{n}^{(k)}}^2} \mathbf{S}_{n}^{(k)})\\
\end{split}
\end{equation}
With  $\gamma_e$ in (\ref{Equ:SNROfCapacityBounds}), we can get the lower-bound of the channel capacity on the $(n,k)$-th subcarrier.

\section{}\label{Appen:ergodicCapacityThm}
The ergodic capacity of i.i.d. Rayleigh fading MIMO channels is given by \cite{zhang2005very} as

\begin{equation}\label{Equ:ErgodicCapacity}
\begin{split}
&\log_2\left( 1 + \sum_{k=1}^n \alpha_k(\gamma) e^{\begin{matrix} \sum_{i=0}^{k-1} \psi (m-i) \end{matrix}} \right)  \leq E(C) \\
&{\qquad \qquad \qquad \qquad} \leq \log_2\left( 1 + \sum_{k=1}^n \alpha_k(\gamma) \prod_{i=1}^{k-1} (m-i) \right),
\end{split}
\end{equation}
wherein, $\alpha _k(\gamma)$ and $\psi(x)$ have the same definition as \emph{lemma} \ref{Lem:ErgodicCapacityLoss} and $E(\cdot)$ is the expectation function.

 By using the ESNR  in \eqref{Equ:SNROfCapacityBounds} to \eqref{Equ:ErgodicCapacity}, one can get the ergodic capacity at the $(n,k)$-th subcarrier is lower-bounded by

\begin{equation*}
E(C_n^{(k)}) \geq \log_2\left( 1 + \sum_{k=1}^n \alpha_k(\gamma_e) e^{\begin{matrix} \sum_{i=0}^{k-1} \psi (m-i) \end{matrix}} \right).
\end{equation*}

Similarly, the ergodic capacity at the $(n_c,k_c)$-th subcarrier is upper-bounded by

\begin{equation*}
E(C_{n_c}^{(k_c)}) \leq \log_2\left( 1 + \sum_{k=1}^n \alpha_k(\gamma) \prod_{i=1}^{k-1} (m-i) \right).
\end{equation*}

Thus, the results in \emph{Lemma} \ref{Lem:ErgodicCapacityLoss} can be easily derived.

\end{appendices}

\bibliography{bibfile}



\end{document}